\documentclass[aps,pra,
floatfix,notitlepage,
nofootinbib,
longbibliography,11pt]{revtex4-1}
\usepackage{amsmath,amsthm,amsfonts,amssymb,braket}
\usepackage{graphicx}
\usepackage{color}

\makeatletter
\def\l@subsubsection#1#2{}
\makeatother

\usepackage[colorlinks=true,citecolor=blue,linkcolor=blue]{hyperref}
\usepackage[capitalize,nameinlink]{cleveref}

\newtheorem{lemma}{Lemma}[section]
\newtheorem{theorem}[lemma]{Theorem}

\theoremstyle{definition}
\newtheorem{definition}[lemma]{Definition}

\newcommand{\CC}{{\mathbb{C}}}

\newcommand{\ZZ}{{\mathbb{Z}}}
\newcommand{\FF}{{\mathbb{F}}}

\newcommand{\cg}[1]{{\phi_\#^{({#1})}}}
\newcommand{\cgx}[1]{{\chi_\#^{({#1})}}}

\DeclareMathOperator{\im}{\mathrm{im}}

\DeclareMathOperator{\coker}{\mathrm{coker}}
\DeclareMathOperator{\ann}{\mathrm{ann}}

\newcommand{\stab}{\mathcal{S}}

\newcommand{\mm}{\mathfrak{m}}

\newcommand{\imag}{\mathrm{i}}

\newcommand{\GL}{\mathsf{GL}}
\newcommand{\Sp}{{\mathsf{Sp}^\dagger}}
\newcommand{\ESp}{{\mathsf{ESp}^\dagger}}

\begin{document}

\title{Classification of translation invariant topological Pauli stabilizer codes 
for prime dimensional qudits on two-dimensional lattices}

\author{Jeongwan Haah}
\affiliation{Microsoft Quantum and Microsoft Research, Redmond, Washington, USA}

\begin{abstract}
We prove that 
on any two-dimensional lattice of qudits of a prime dimension,
every translation invariant Pauli stabilizer group
with local generators and with code distance being the linear system size,
is decomposed by a local Clifford circuit of constant depth 
into a finite number of copies of the toric code stabilizer group (abelian discrete gauge theory).
This means that under local Clifford circuits
the number of toric code copies is the complete invariant of topological Pauli stabilizer codes.
Previously, the same conclusion was obtained
under the assumption of nonchirality for qubit codes
or the Calderbank-Shor-Steane structure for prime qudit codes;
we do not assume any of these.
\end{abstract}

\maketitle

\section{Introduction and Result}

Topological phases of matter of gapped systems refer to equivalence classes of physical systems
where two systems are deemed equivalent if they can be smoothly deformed to each other
while preserving the energy gap above the ground state subspace.
By definition, any characteristic of a topological phase of matter 
should be robust under any local perturbations.
The study of topological phases of matter thus focuses, to a large extent, 
on identifying invariants of systems under smooth deformations
and on concrete models that exhibit different invariants.
Indeed, many sophisticated model Hamilonians
are constructed to date in various dimensions in the absence or presence of symmetries~\cite{dijkgraaf1990topological,Levin_2005,WalkerWang}.
However, the classification (i.e., finding complete invariants) of topological phases 
remains still an active problem.

Here we present a classification theorem on
translation-invariant topological Pauli stabilizer code Hamiltonians in two spatial dimensions.
These are a class of Pauli stabilizer codes~\cite{CalderbankShor1996Good,Steane1996Multiple,Gottesman1996Saturating,CalderbankRainsShorEtAl1997Quantum}
whose stabilizer group generators define a local Hamiltonian
that exhibits intrinsic topological order~\cite{Wen1991SpinLiquid,Kitaev_2003}.
By construction, the Hamiltonian terms are commuting,
and all operators that determine the topological invariants
are tensor products of Pauli matrices.
Simply put, our result is that there is only the toric code (abelian discrete gauge theory)
in two dimensions.

Though Pauli stabilizer models are simple to calculate about,
some Hamiltonians in this class realize exotic phases of matter 
(e.g. fracton topological order~\cite{VHF2016})
in three or higher spatial dimensions.
This makes our conclusion in two dimensions more interesting;
later we will highlight the steps in the proof of our result 
where it is important to work in two dimensions.

Before we present our rigorous statement,
let us review previous reasoning that has alluded our result.
Believing in the effective description by unitary modular tensor categories (UMTC)~\cite{BakalovKirillov2001,Kitaev_2005},
we can tabulate possible topological phases of matter 
realized by two-dimensional topological Pauli stabilizer codes.
Indeed, if we consider generalized Pauli operators (or discrete Weyl operators)
over a system of qudits of dimension $p$, which we will define below,
the topological spins for a Pauli stabilizer code model must be valued in $p$-th roots of unity
and the UMTC is determined by a quadratic form $\theta$ over $\ZZ/p\ZZ$; 
see e.g.~\cite[\S 5]{Galindo2016} and references therein.
When $p$ is a prime (so that the topological spins are valued in a finite field $\FF_p$)
nondegenerate quadratic forms are particularly simple~\cite{Lam,Kniga}.
If we ignore direct summands of hyperbolic forms, 
which correspond to the toric code phase~\cite{Kitaev_2003},
then the only nontrivial possibilities are listed as follows.
Let $v$ denote a general vector in the domain $V$ (an $\FF_p$-vector space)
of a quadratic form $\theta: V \to \FF_p$.
Recall that an $n$-dimensional quadratic form over $\FF_p$ 
is a homogeneous polynomial in $n$ variables of total degree~2
with coefficients in $\FF_p$.

\begin{enumerate}
\item ($p=2$) a two-dimensional form $\theta_{3F}(v) = v_1^2 + v_1 v_2 + v_2^2$
which corresponds to the ``three-fermion'' theory.
\item ($p \equiv 3 \mod 4$ so the Witt group is $\ZZ_4$) 
a one-dimensional form $\theta_1(v) = v^2$,
its time-reversal conjugate $-\theta_1$, 
or a direct sum $\theta_1 \oplus \theta_1 \cong (-\theta_1) \oplus (-\theta_1)$.
\item ($p \equiv 1 \mod 4$ so the Witt group is $\ZZ_2 \oplus \ZZ_2$)
a one-dimensional form $\theta_1(v) = v^2$,
another one-dimensional form $\theta_{\alpha}(v) = \alpha v^2$ where $\alpha$ is any nonsquare element of $\FF_p$,
or their direct sum $\theta_1 \oplus \theta_\alpha$.
\end{enumerate}

\noindent
This motivates us to ask whether these candidates are exhaustive
and whether every candidate can be realized in a lattice model.
For Pauli stabilizer code models the problem becomes 
a question on locally generated abelian multiplicative groups of Pauli operators.

Previously,
Bomb\'in~\cite{Bombin2011Structure} has studied topological stabilizer groups in two dimensions 
that are translation invariant
on systems of qubits ($p=2$) and concluded that, 
up to locality preserving automorphisms of operator algebra%
\footnote{
Locality preserving automorphisms (also called quantum cellular automata)
are not necessarily shallow quantum circuits.
It is only recent results~\cite{FreedmanHastings,haah2019clifford}
that such a locality preserving automorphisms in two dimensions is actually a shallow quantum circuit.
} 
that map every Pauli matrix to a product of Pauli matrices~\cite{clifQCA,haah2019clifford},
any such group is a direct sum of copies of the toric code stabilizer group
and a trivial stabilizer group for a product state 
\emph{if} the topological charge content is \emph{nonchiral},
i.e., its decomposition does not contain the three-fermion theory.
Any chiral theory in two dimensions is widely believed to be not realizable by
any commuting Hamiltonian,
but there are at least two definitions of chirality in this statement:
one from thermal energy current~\cite[App.D]{Kitaev_2005} 
and the other from algebraic theory of anyons~\cite[cf.Prop.6.20]{FroehlichGabbiani}.
We are unaware of formal connections between these notions.

In relation to \cite{Bombin2011Structure}, 
the present author has studied~\cite{Haah2016} a similar translation invariant case with qudits of a prime dimension
under the assumption that the group be generated by tensor products of $X$ and tensor products of $Z$ 
(the Calderbank-Shor-Steane structure~\cite{CalderbankShor1996Good,Steane1996Multiple})
and concluded that every such stabilizer group is a direct sum of those for the toric code
after a geometrically local circuit of control-Not gates.

Recently a new ingredient was obtained~\cite{HFH2018},
which proves that every translation invariant topological Pauli stabilizer code on prime qudits
must have a nontrivial ``boson.''
A boson is by definition an excitation with topological spin $1 \in \CC$
(or $0 \in \ZZ/p\ZZ$ in our setting).
When topological spin corresponds to a quadratic form,
the existence of a boson means that the quadratic form~$\theta$ 
has a nonzero vector whose value is zero,
in which case the quadratic form is called \emph{isotropic}.
Hence, it rules out all the nontrivial possibilities listed above
where we have listed all anisotropic quadratic forms over prime finite fields.
This has left a question of whether the equivalence of quadratic forms of topological spins
implies the equivalence of the Pauli stabilizer groups under Clifford circuits;
the latter is stronger than just the equivalence of the corresponding topological phases of matter.
In this paper we answer this question in the affirmative,
closing the classification problem of two-dimensional topological Pauli stabilizer codes
over prime dimensional qudits up to Clifford circuits in the translation invariant case.
Let us state our result precisely.

We first recall standard definitions.
A {\bf generalized Pauli matrix} (discrete Weyl operators)
for a $p$-dimensional qudit $\CC^p$ is any product of
\begin{align}
\exp\left(\frac{2\pi \imag}{p}\right) I, \quad
X = \sum_{j \in \ZZ_p} \ket{j +1}\bra{j} \quad\text{ and }\quad 
Z = \sum_{j \in \ZZ_p} \exp\left(\frac{2\pi \imag}{p} j \right)\ket j \bra j .
\end{align}
These are defined for any integer $p \ge 1$,
but in this paper $p$ is always a prime.
A {\bf generalized Pauli operator} is any \emph{finite} tensor product of generalized Pauli matrices.
A {\bf Clifford gate} is a finitely supported unitary 
that maps every generalized Pauli operator to a generalized Pauli operator,
and a {\bf Clifford circuit} is a finite composition of layers of nonoverlapping Clifford gates
each of which is supported on a ball of a uniformly bounded radius.
The number of layers in a circuit is the {\bf depth} of the circuit.
A {\bf trivial stabilizer group} is the group of generalized Pauli operators
generated by $Z$ on every qudit.
The {\bf toric code stabilizer group}~\cite{Kitaev_2003}
is an abelian group of generalized Pauli operators
on a two-dimensional lattice $\ZZ^2$ with two $p$-dimensional qudits per lattice point,
generated by $X_{s+\hat y,1} X_{s,1}^\dagger X_{s+\hat x,2}^\dagger X_{s,2}$
and $Z_{s-\hat x,1}^\dagger Z_{s,1} Z_{s-\hat y,2}^\dagger Z_{s,2}$
for all sites $s$.

\begin{theorem}\label{thm:main}
Let $\stab$ be an abelian group of generalized Pauli operators acting on a two-dimensional square lattice $\ZZ^2$
with $q \ge 1$ qudits of a prime dimension $p$ per lattice point. Suppose that
\begin{enumerate}
\item (No frustration) if $\omega I \in \stab$ for $\omega \in \CC$, then $\omega = 1$,
\item (Translation invariance up to a phase) if $P \in \stab$, then for every translate $P'$ of $P$, 
we have $\omega P' \in \stab$ for some $p$-th root of unity $\omega \in \CC$, and
\item (Topological order) if a generalized Pauli operator $P$ commutes with every operator of $\stab$, then $\omega P \in \stab$ for some $\omega \in \CC$.
\end{enumerate}
Then, there exists a Clifford circuit of finite depth that maps $\stab$ into
a direct sum $\mathcal T^{\oplus n} \oplus \mathcal Z$ for some $n \ge 0$
where $\mathcal T$ is the toric code stabilizer group,
and $\mathcal Z$ is the trivial stabilizer group.
Here the circuit is translation invariant
with respect to a subgroup of the full translation group $\ZZ^2$ with finite index.
\end{theorem}

Physically, the stabilizer group as a whole is more important than a generating set
since any local generating set can be used to define a gapped Hamiltonian
but the quantum phase of matter only depends on the group~\cite[\S 2]{Haah2013}.
This is why we have stated our theorem in terms of groups.
Our result should be understood as the scope of the topological phases
that can be realized by unfrustrated commuting Pauli Hamiltonians.

One might wonder why there is no reference to the length scale of the topological order.
This is because the definition of topological order in \cite{BravyiHastingsMichalakis2010stability}
is for a family of finite lattices,
whereas we work with infinite lattices.
Our assumption implies the topological order condition of \cite{BravyiHastingsMichalakis2010stability}
with length scale being the linear system size 
under periodic boundary conditions~\cite{Haah2013}.
Any implicit finite number in the theorem such as the depth of the Clifford circuit
and the index of the translation subgroup
gives a constant uniform bound on a corresponding quantity
for an infinite family of finite systems.

Note that our theorem assumes a finite dimensional degrees of freedom per site.
This is not just a technical convenience but rather a fundamentally important assumption.
Indeed, in a limit $p \to \infty$ that would produce a rotor $U(1)$ model,
the basic result that any excitation is attached to a string operator and hence is mobile,
cease to be true in general; consider 2+1D analogues of \cite{Pretko2016}.
Moreover, when it comes to rotor models 
we do not have any stability result against perturbations such as \cite{BravyiHastingsMichalakis2010stability}
but only an instability result~\cite{Polyakov1975}.

It remains an open problem to relax the translation invariance.
One might be able to promote an arbitrary system to a periodic system~\cite{Hastings2013},
but it appears difficult to adapt the present classification proof directly
to nonperiodic situations;
the argument for the existence of a boson in \cite{HFH2018} 
relies on bilinear forms over a field of fractions in one variable,
for which the translation invariance is used.

Even assuming translation invariance,
there is a quantitative question left.
In our mapping from a given stabilizer group to a direct sum of toric code stabilizer groups,
we had to break the translation invariance down to a smaller group.
To some extent this is necessary;
a translation group may act nontrivially on the fusion group of anyons~\cite{Wen2003Plaquette},
while for the toric code the action is trivial.
However, almostly surely, our choice of smaller translation group is not the biggest possible;
in fact, we will not keep track of the index of this translation subgroup.
One can then ask what the precise order (exponent) 
of the translation group action on the fusion group of anyons is.
Generally this question should be answered as a function of 
interaction range and unit cell size.
A lower bound that is exponential in the interaction range is known~\cite[\S 7 Rem.~3]{Haah2013},
but the upper bound is largely open.

The rest of this paper constitutes the proof of \cref{thm:main}.

\section{Transcription to polynomials}

Following \cite{Haah2013} (see also \cite{Haah2016} and a summary section \cite[\S IV.A]{HFH2018})
we transcribe the problem into a polynomial framework
by regarding translation invariant groups of generalized Pauli operators modulo phase factors,
which are abelian, 
as modules over the translation group algebra
\begin{align}
R = \FF_p[x^\pm, y^\pm].
\end{align}
In particular, the abelianized group of generalized Pauli operators is a free module $R^{2q}$
where $q$ is the number of qudits per lattice site (unit cell),
equipped with a nondegenerate symplectic form that captures commutation relations.
Below we will not distinguish generalized Pauli operators from an element of $R^{2q}$
as the phase factors will not be important.

Already this perspective implies the following:
Note that in \cref{thm:main} we do not assume that the stabilizer group is generated
by operators whose supports are contained in disks of a uniformly bounded radius;
that is, we do not assume that the stabilizer group has a local generating set.
However, the supposition of \cref{thm:main} implies
that the stabilizer group has a local generating set.
The group of all generalized Pauli operators (each of which is finitely supported)
up to phase factors is $R^{2q}$, which is a finitely generated module over $R$.
The stabilizer group modulo phase factors is a subset of $R^{2q}$
but the translation invariance implies that this subset is actually a submodule.
Since $R$ is Noetherian, the stabilizer submodule must be finitely generated,
which means that the stabilizer group does have local generators
with uniformly bounded support size,
where the bound is given by a finite generating set for the stabilizer submodule.

We use $\FF_p$-linear ring homomorphisms $\phi^{(m)} : \FF_p[x'^\pm, y'^\pm] \to R$ 
such that $x' \mapsto x^m$ and $y' \mapsto y^m$
indexed by positive integer $m$
to denote {\bf coarse-graining},
which induces formally a covariant functor $\cg{n}$ on the category of $R$-modules.
The domain of this morphism is interpreted as the group algebra for a smaller translation group,
enlarging the unit cell of the qudit system $m \times m$ times as large as the original one.
We use $\overline \cdots$ to denote the $\FF_p$-linear {\bf involution} of $R$ 
such that $x \mapsto \bar x = x^{-1}$ and $y \mapsto \bar y = y^{-1}$,
and $\dagger$ the involution followed by transpose for matrices over $R$.
Let $I_q$ denote the $q \times q$ identity matrix.
We fix specific matrices over $R$:
\begin{align}
\lambda_q = \begin{pmatrix} 0 & I_q \\ - I_q & 0 \end{pmatrix}, \quad
\epsilon_0 = \begin{pmatrix}
x-1 & y-1 & 0 & 0 \\
0   & 0   & \bar y -1 & -\bar x + 1
\end{pmatrix} \quad \text{ and }\quad
\sigma_0 = (\epsilon_0 \lambda_q^{-1})^\dagger
 \label{eq:eps0}
\end{align}
where $\epsilon_0$ describes the $\ZZ_p$ toric code~\cite{Kitaev_2003} 
on the square lattice~\cite[\S5 Ex.~2]{Haah2013}.
For the clarity in notation we define a matrix $E_{i,j}(a)$ for any $a \in R$ as
\begin{align}
 \left[ E_{i,j}(a) \right]_{\mu \nu} &= \delta_{\mu \nu} + \delta_{\mu i} \delta_{\nu j} a 
 & \text{ where $\delta$ is the Kronecker delta}.
\end{align}

\begin{definition}\label{def:ESp}
For a given positive integer $q$,
the following $2q \times 2q$ matrices generate the 
{\bf elementary symplectic group} denoted by $\ESp(q;R)$.
Below $a \in R^\times$ is any monomial.
\begin{align}
 \text{Hadamard: } \quad&E_{i,i+q}(-1) E_{i+q,i}(1) E_{i,i+q}(-1) &\text{ where } 1 \le i \le q,\nonumber\\
 \text{control-Phase: }\quad & E_{i+q,j}(a)E_{j+q,i}(\bar a)&\text{ where } 1 \le i,j \le q,\nonumber \\
 \text{control-Not: } \quad & E_{i,j}(a) E_{j+q,i+q}(-\bar a) &\text{ where }1 \le i \ne j \le q, \nonumber \\
\text{extra gate: } \quad & E_{i,i}(a-1) E_{i+q,i+q}(a^{-1}-1) & \text{ where } 1 \le i \le q. \nonumber
\end{align}
The {\bf symplectic group} denoted by $\Sp(q;R)$ consists of 
all $U \in \mathrm{Mat}(2q;R)$ such that $U^\dagger \lambda_q U = \lambda_q$.%
\footnote{
The $\dagger$ in the notation $(\mathsf{E})\Sp$ refers to the fact that
we collect matrices $U$ that obeys $U^\dagger \lambda_q U = \lambda_q$.
Sometimes a different group consisting of $U$ such that $U^T \lambda_q U = \lambda_q$
is considered in literature, and our notation differentiate them.
}
\end{definition}
Note that the control-Phase%
\footnote{
This includes the induced action by the ``phase'' gate 
$\mathrm{diag}(1,\imag)$ but not $\mathrm{diag}(1,1,1,\imag)$ for qubits ($p=2$).
Hence, this terminology may be inconsistent with that elsewhere.
}
gate with $i = j$ is equivalent to $E_{i+q,i}(f)$ for some $f = \bar f$,
and conversely any $E_{i+q,i}(f)$ with $f = \bar f$ can be written as the control-Phase with $i=j$
since any such $f$ is of form $a + \bar a$ for some nonunique $a$.
The extra gate is the identity if $p=2$.

The elementary symplectic group is a subgroup of the symplectic group as one can verify,
but they are not equal if $R$ was a Laurent polynomial ring of three or more variables~\cite{haah2019clifford}.
If~$R = \FF_p[x^\pm, y^\pm]$, the Laurent polynomial ring of \emph{two} variables,
then for every $U \in \Sp(q;R)$ there exists a positive integer $m$ such that $\cg{m}(U) \in \ESp(m^2 q;R)$~\cite{haah2019clifford}.
That is, every symplectic transformation on two-dimensional lattice is a Clifford circuit with weaker translation invariance.
Hence, in the present paper where we only have two-dimensional lattices,
we do not have to distinguish $\ESp$ and $\Sp$.

\begin{theorem}[Implying \cref{thm:main}]\label{thm:poly}
Let $\sigma$ be a $2q \times t$ matrix over $R$,
interpreted as a map acting on the left of a column vector of length $t$,
such that
\begin{align}
\ker \sigma^\dagger \lambda_q  = \im \sigma . \label{eq:exact}
\end{align}
Then, there exist an integer $m \ge 1$ and a matrix $E \in \ESp(m^2q;R)$ such that
\begin{align}
E \im \cg{m}(\sigma) = \left(\im \begin{pmatrix} I_{m^2q - k} \\ 0 \end{pmatrix} \right) \oplus \bigoplus^{k/2} \im \sigma_0 
\end{align}
where the images are over $R' = \FF_p[x^{\pm m}, y^{\pm m}] \subseteq R$
and $k$ is an even integer that is equal to the $\FF_p$-dimension of the torsion submodule of $\coker \sigma^\dagger$.
\end{theorem}

\begin{proof}[Proof of the transcription]
Any translation invariant abelian group of generalized Pauli operators
corresponds to a submodule $S$ of the $R$-module $P$ 
of all generalized Pauli operators (forgetting phase factors) with the property that 
$v^\dagger \lambda_q v = 0$ for any $v \in S$~\cite[Prop.~1.2]{Haah2013}.
Picking a generating set for the module $S$,
which amounts to writing $S = \im \sigma$ for some matrix $\sigma$ over $R$,
we have a matrix equation $\sigma^\dagger \lambda_q \sigma = 0$.
The topological order condition~\cite{BravyiHastingsMichalakis2010stability}
for unfrustrated commuting Pauli Hamiltonians
is equivalent to $\ker \sigma^\dagger \lambda_q = \im \sigma$~\cite[Lem.~3.1]{Haah2013}.
An elementary symplectic trasformation 
is always induced by a Clifford circuit of finite depth~\cite[\S2.1]{Haah2013}.
These symplectic transformations are forgetful only of a conjugation by a 
(possibly infinitely supported) generalized Pauli operator~\cite[Prop.~2.2]{Haah2013}.
The module $\im \sigma_0$ is the same as the toric code stabilizer group.
\end{proof}

\section{Proof}\label{sec:proof}

The proof of \cref{thm:poly} is by induction 
in the vector space dimension $k$ of the torsion part of $\coker{\epsilon}$
where $\epsilon = \sigma^\dagger \lambda_q$ is the excitation map of the given stabilizer group.
It has been proved that the excitation map $\epsilon$ can be chosen 
(as it depends on a generating set for the stabilizer group)
such that $\coker{\epsilon}$ is a torsion module
if $R$ is the Laurent polynomial ring with two variables,
and furthermore after a suitable choice of a smaller translation group (coarse-graining)
the annihilator of $\coker \epsilon$ becomes a maximal ideal 
$\mm = (x-1,y-1) \subset R$~\cite[\S 7 Thm.~4]{Haah2013}.
Such $\epsilon$ must be $q \times 2q$ to satisfy the topological order condition
$\ker \epsilon = \im \lambda_q \epsilon^\dagger$.
Hence, without loss of generality we begin with an extra assumption,
whenever $\coker \epsilon \neq 0$, that
\begin{align}
\ann \coker \epsilon = \mm = (x-1,y-1). \label{eq:annmm}
\end{align}
This implies that $k$ is finite.
We will extract a copy of the toric code whenever $k \ge 2$,
decreasing this dimension $k$ by 2.
This is done by trivializing string operators that transports bosons
whose existence is proved in \cite[Cor.~III.20]{HFH2018}.
The possibility $k=1$ will be ruled out in the course of the proof.
The case $k=0$ is treated in \cite[\S IV.B]{HFH2018}
where a Clifford quantum cellular automaton is constructed 
to map the stabilizer group to the trivial stabilizer group.
But it is shown that any Clifford QCA in two dimensions 
is in fact a Clifford circuit with weaker translation invariance 
up to a shift~\cite{haah2019clifford}.
A shift QCA does not change the stabilizer group at all.
These will complete the proof.

\subsection{Topological spin}

\begin{figure}
\centering
\includegraphics[width=0.8\textwidth, trim = {0ex 75ex 95ex 0ex}, clip]{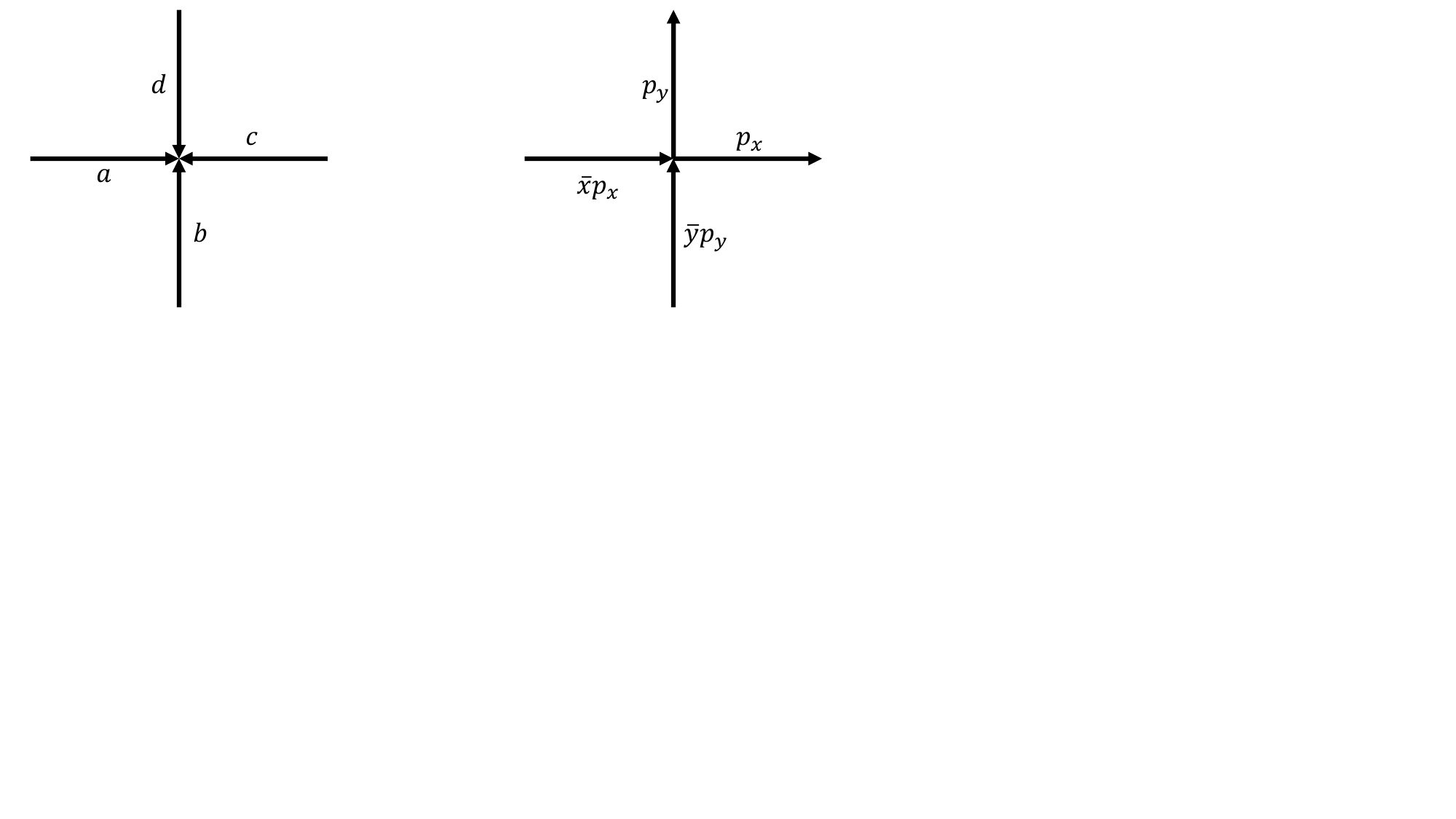}
\caption{
Left: 
Movers (hopping operators) that send a topological charge $e$ to the center.
Commutation relations of \emph{any} triple of the movers determine the topological spin $\theta$ of $e$;
\cref{lem:thetawelldefined}.
Right:
For a boson ($\theta = 0$) we can choose movers $p_x,p_y$ on a coarse scale
such that they always commute;
\cref{lem:isotropy-commutingmover}.
}
\label{fig:topspin}
\end{figure}

We consider the topological spins of the excitations of a given topological stabilizer group.
We will define them using hopping operators for the excitations 
in a way that is tailored to our setting.
Recall that two generalized Pauli operators represented as two column vectors $u,v$ over $R$,
commute with each other if and only if
\begin{align}
[u,v] = - [v,u] \in \FF_p
\end{align}
vanishes,
where $[u,v]$ is
the coefficient of $x^0 y^0 = 1 \in R$ in $u^\dagger \lambda_q v$~\cite[Prop.1.2]{Haah2013}.
Let us introduce a notation for any  $a,b,c \in R^{2q}$
\begin{align}
[a,b,c] &= [a, b] + [b, c] + [c, a] .
\end{align}
It is obvious that $[a,b,c] = [b,c,a] = [c,a,b]$.

\begin{definition}\label{def:tspin}
Assuming \cref{eq:annmm},
an {\bf excitation} is any element in the codomain of $\epsilon$
or equivalently any element in the domain of $\sigma$.
An excitation $e$ is {\bf nontrivial} if $e \notin \im \epsilon$.
A {\bf (topological) charge} is then an equivalence class of excitations modulo trivial ones.
An {\bf $x$-mover} $p_x(e)$ for an excitation $e$
is any generalized Pauli operator $p \in R^{2q}$ such that $\epsilon( p ) = (x-1)e$.
Likewise, a {\bf $y$-mover} $p_y(e)$ is any $p \in R^{2q}$ such that $\epsilon (p) = (y-1)e$.
The {\bf topological spin} $\theta(e) \in \FF_p$ of an excitation is
\begin{align}
\theta(e) &= \lim_{n \to +\infty} [a_n,b_n,c_n] \label{eq:theta}\\
\text{where} \quad
a_n &= (x^{-n} + x^{-n+1} + \cdots + x^{-1})p_x(e), \nonumber\\
b_n &= (y^{-n} + y^{-n+1} + \cdots + y^{-1})p_y(e), \quad \text{ and } \nonumber\\
c_n &= -(1 + x + x^2 + \cdots + x^{n-1}) p_x(e). \nonumber
\end{align}
\end{definition}

The topological spin can and should be defined much more generally,
but we use this narrow definition which is sufficient for this paper.
The existence of movers for any excitation is precisely the content of \cref{eq:annmm}.
Note that movers are not unique 
since any element of $\ker \epsilon = \im \sigma$ can be added to them.
Our definition of topological spin is due to \cite{LevinWen2003Fermions}
using commutation relations among three hopping operators (movers)
attached to an excitation; we have transcribed it into our additive notation.
The limit as $n \to \infty$ exists because $a_n,b_n,c_n$ are 
long ``string'' operators and string segments far away 
from the origin of the lattice always commute
and hence the sequence in~$n$ becomes constant eventually.
Pictorially, $a_n$ is inserting $e$ from the left infinity to the origin,
$b_n$ from the bottom, and $c_n$ from the right.
See \cref{fig:topspin}.

\begin{lemma}\label{lem:thetawelldefined}
First, $\theta(e)$ is independent of the choices of the movers.
Second, $\theta(e') = \theta(e)$ if $e'-e$ is trivial.
Third, $\theta(e) = \theta(\cg{m}(e))$ for any positive integer $m$.
Fourth, $\theta(e) = [a_n,b_n,d_n] = [a_n,c_n,d_n] = [b_n,c_n,d_n]$ as $n \to \infty$
where
\begin{align}
d_n &= -(1+ y + y^2 + \cdots +y^{n-1})p_y(e).
\end{align}
\end{lemma}
The fourth claim is basically that the triple commutator can consist of
any string operators as long as they circle around the excitation~$e$ counterclockwise.
\begin{proof}
First, 
the movers are unique up to $\ker \epsilon = \im \sigma$,
and hence $\tilde p_{x,y}(e) = p_{x,y}(e) + s_{x,y}$ for some $s_x,s_y \in \im \sigma$ 
is the most general form of movers.
Let us examine
\begin{align}
\tilde \theta(e) - \theta(e) = 
\lim_{n\to \infty}
[\tilde a_n - a_n, b_n - \tilde c_n] + 
[\tilde b_n - b_n, c_n - \tilde a_n] + 
[\tilde c_n - c_n, a_n - \tilde b_n]
\end{align}
where those with tilde are by the modified movers $\tilde p_{x,y}(e)$,
and show that each term vanishes on its own.
The first term is $[(x^{-n} + x^{-n+1} + \cdots + x^{-1})s, b_n - \tilde c_n]$ 
for some sufficiently large $n$.
The difference $b_n - \tilde c_n$ creates excitations near $(0,-n) \in \ZZ^2$ and $(n,0) \in \ZZ^2$ where
$(x^{-n} + x^{-n+1} + \cdots + x^{-1})s$ does not have any support,
and hence the commutator vanishes.
The other two terms vanish for completely parallel reasons.

Second,
let $g \in R^{2q}$ be a finitely supported operator such that $\epsilon(g) = e' - e$.
By the first claim we can choose movers for $e'$ as $p_x(e') = p_x(e) + (x-1)g$ 
and $p_y(e') = p_y(e) + (y-1)g$.
Then,
$a_n(e') - a_n(e) = (1-x^{-n})g$,
$b_n(e') - b_n(e) = (1-y^{-n})g$,
and $c_n(e') - c_n(e) = (1-x^{n+1})g$.
Now we see
\begin{align}
\theta(e') - \theta(e) 
&=
[g-x^{-n}g,b_n] + [a_n, g - y^{-n}g] + [g-x^{-n}g,g - y^{-n}g]\nonumber\\
&\quad+
[g - y^{-n}g,c_n] + [b_n,g - x^{n+1}g] + [g - y^{-n}g,g - x^{n+1}g]\nonumber\\
&\quad+
[g - x^{n+1}g,a_n] + [c_n,g-x^{-n}g] + [g - x^{n+1}g,g-x^{-n}g]\nonumber\\
&=
[g, b_n] + [a_n, g] + [g,g]\nonumber\\
&\quad +
[g,c_n] + [b_n,g] + [g,g]\\
&\quad +
[g,a_n] + [c_n,g] + [g,g]\nonumber\\
&= 0\nonumber
\end{align}
where 
the first equality is by linearity of $[\cdot,\cdot]$,
the second equality is because operators of disjoint support commute 
(e.g. $[x^{-n}g,b_n]=0$ for all sufficiently large $n$),
and the third equality is because $[g,g]=0$ and $[u,v]=-[v,u]$ for any $u,v\in R^{2q}$.

Third, if we have movers $p_x,p_y$ for $e$,
we can choose movers $p_{x'}, p_{y'}$ for $\cg{m}(e)$
as $p_{x'} = \cg{m}((1+x+\cdots+ x^{m-1})p_x)$ 
and $p_{y'} = \cg{m}((1+y+\cdots+y^{m-1})p_y)$.
The commutator is oblivious to $\cg{m}$ 
and the new movers are just longer (larger $n$) string operators.

Fourth, we examine the difference, omitting the subscript $n$ for brevity:
\begin{align}
[a,b,c] - [a,b,d]
&=
[a,b]+[b,c]+[c,a] - [a,b]-[b,d]-[d,a] \nonumber \\
&= [b-a,c-d] .
\end{align}
The operator $b-a$ creates excitations at $(0,-n),(-n,0) \in \ZZ^2$
while $c-d$ is supported around the positive $x$ and $y$ axes.
Hence, we can find a straight diagonal string operator $f$ connecting $(0,-n)$ and $(-n,0)$
such that $[f,c-d] =0$ (for their disjoint support) and $\epsilon(b-a+f) = 0$.
Since $\ker \epsilon = \im \sigma$,
we know $b-a+f \in \im \sigma$.
Similarly, we can find a straight diagonal string operator $h$ 
connecting $(n,0) \in \ZZ^2$ and $(0,n) \in \ZZ^2$
such that $c-d+h \in \ker \epsilon = \im \sigma$ and $[b-a+f,h ] = 0$.
This means that $[b-a,c-d] = [b-a+f,c-d+h]$, but the latter is zero 
since $[u,v]=0$ for any $u,v \in \im \sigma$.
This proves that $\theta(e) = \lim_{n \to\infty}[a_n,b_n,c_n] = \lim_{n \to \infty} [a_n,b_n,d_n]$.
Similarly,
\begin{align}
[a,b,d] - [a,c,d] 
&= [a,b]+[b,d]+[d,a] - [a,c] - [c,d] - [d,a]\nonumber\\
&= [a-d,b-c]
\end{align}
The operator $a-d$ creates excitations at $(-n,0),(0,n) \in \ZZ^2$
and the operator $b-c$ creates excitations at $(0,-n), (n,0) \in \ZZ^2$,
which can be turned into elements of $\im \sigma$ 
by ``capping off'' by straight diagonal string operators
and the commutator $[a-d,b-c]$ is shown to vanish.
A parallel argument to shows that $[a_n,c_n,d_n] = [b_n,c_n,d_n]$ for all sufficiently large $n$. 
\end{proof}

Therefore, $\theta$ is well defined on the set of all topological charges
(equivalence classes of excitations).
In fact, $\theta$ is a quadratic form on the $\FF_p$-vector space $\coker \epsilon$.%
\footnote{
The definition of a quadratic form~\cite{Kniga}
requires that its polar form $S(e,e') = \theta(e+e')-\theta(e)-\theta(e')$ be bilinear over $\FF_p$.
This is true, and in fact one can show that this modular $S$ is precisely
the commutation relation of all logical operators of the finite dimensional stabilizer code
on any sufficiently large but finite 2-torus.
In particular, $S$ is nondegenerate.
One can classify the quadratic forms $\theta$ based only on the fact that it is $\FF_p$-valued
and $S$ is nondegenerate, and the result is in Introduction.
}
If $\theta(e) = 0$, we call $e$ a {\bf boson}.
\begin{lemma}[Cor.~III.20 of \cite{HFH2018}]\label{lem:isotropic}
If $\coker \epsilon \neq 0$, 
then there exists a nontrivial excitation $e \notin \im \epsilon$ such that $\theta(e) = 0$,
i.e., a nontrivial boson exists.
\end{lemma}

\subsection{Simplification of movers and excitation maps}

The rest of the proof of \cref{thm:poly} consists of elementary 
--- a little complicated but not difficult --- computation.
In \cref{lem:isotropy-commutingmover} below
we will choose movers for a nontrivial boson such that they commute with any of its translates.
In \cref{lem:commutingmover=trivialmover}
we will further simplify the movers to determine two columns of $\epsilon$.
Then, using the exactness condition \cref{eq:exact},
we will find an accompanying row of $\epsilon$;
this will rule out the possibility $k := \dim_{\FF_p} \coker \epsilon = 1$.
The determined columns and row will be further simplified in \cref{lem:xyab}
and in turn will single out a direct summand $\epsilon_0$ from $\epsilon$.

\begin{lemma}\label{lem:isotropy-commutingmover}
Under a choice of a sufficiently large unit cell
(reducing the translation group to a subgroup of a finite index via $\cg{m}$ for some $m$),
the movers $p_{x'}(e), p_{y'}(e)$ for a boson $e$ 
can be chosen such that 
$p_{x'}(e)^\dagger \lambda_{m^2q} p_{y'}(e) = p_{x'}(e)^\dagger \lambda_{m^2 q} p_{x'}(e) = p_{y'}(e)^\dagger \lambda_{m^2q} p_{y'}(e) = 0$.
Furthermore, the entries of $p_{x'}(e)$ as a Laurent polynomial vector over $R' = \FF_p[x'^\pm,y'^\pm]$
are $\FF_p$-linear combinations of $1,x'$,
and those of $p_{y'}(e)$ are $\FF_p$-linear combinations of $1,y'$.
\end{lemma}
That is, the movers and all their translates can be made commuting.
\begin{proof}
Let us drop the reference to $e$ since we fix $e$.
Choose arbitrary movers $p_x,p_y$ with respect to $R$; 
$\epsilon(p_x) = (x-1)e$ and $\epsilon(p_y) = (y-1)e$.
If $R' = \FF_p[x'^\pm, y'^\pm]$ injects into $R$ by $\phi^{(m)}: x' \mapsto x^m , y' \mapsto y^m$, then
a mover $p_{x'}$ with respect to $R'$, i.e., $(x'-1)e = (\cg{m}\epsilon)(p_{x'})$,
can be chosen as 
\begin{align}
p_{x'} = s_{tail}+(1+x+\cdots+x^{m-1})p_x + s_{head},
\end{align}
where we did not write $\cg{m}$ for brevity.
The mover $p_{x'}$ (at a coarser lattice) 
is the $m$ movers $p_x$ aligned along the moving direction 
with its ``head'' near $(m,0)\in \ZZ^2$ and ``tail'' near the origin modified
by $s_{head,tail} \in \im \sigma$.
Similarly, 
\begin{align}
p_{y'} = t_{tail} + (1+y+\cdots+y^{m-1})p_y 
\end{align}
with $t_{tail} \in \im \sigma$ near the origin.
(Note that $t_{head}$ is absent.)
We wish to choose these head and tail modifiers such that
\begin{align}
[\overline{x'}p_{x'}, \overline{y'} p_{y'}] &= 0 & \text{ by some } s_{head}, \nonumber \\
[\overline{y'} p_{y'}, - p_{x'}] &= 0 &\text{ by some } s_{tail},\label{eq:choices}\\
[-p_{y'},\overline{x'}p_{x'}] &= 0 &\text{ by some } t_{tail}. \nonumber
\end{align}
These are equations over $\FF_p$ with unknowns $s_{head}, s_{tail}, t_{tail}$ 
and can be solved by inspecting one by one in order.
The first equation contains neither $s_{tail}$ nor $t_{tail}$
as they are far away from the origin 
at which $\overline{x'} p_{x'}$ and $\overline{y'} p_{y'}$ may intersect.
Likewise, the second equation contains neither $s_{head}$ nor $t_{tail}$.
Since a mover must have a nonzero commutator 
with some element of $\im \sigma$ near its head and tail,
we conclude that there exist $s_{head}$ and $s_{tail}$ satisfying the first two equations.
The third equation contains the two unknowns $t_{tail}$ and $s_{head}$,
but given $s_{head}$ determined we can certainly solve it for $t_{tail}$.
For a large enough $m$, the movers $p_{x'}$ and $p_{y'}$ 
are long string operators, and the choices of head and tail modifiers 
are independent of all sufficiently large $m$.

The sequence of \cref{eq:theta} in $n$ constructed by $p_{x'}$ and $p_{y'}$ 
is already constant starting from $n=1$.
Since $e$ is a boson, the fourth claim of \cref{lem:thetawelldefined} implies that
\begin{align}
0 &= [\overline{x'}p_{x'}, \overline{y'} p_{y'}]
 + [\overline{y'} p_{y'}, - p_{x'}]
 + [-p_{x'},\overline{x'}p_{x'}] ,\nonumber\\
0 &= [\overline{x'}p_{x'}, \overline{y'} p_{y'}] 
+ [\overline{y'} p_{y'}, -p_{y'}] 
+ [-p_{y'},\overline{x'}p_{x'}] , \label{eq:bosonimplies}\\
0 &= [\overline{x'}p_{x'},-p_{x'}] + [-p_{x'},-p_{y'}] + [-p_{y'}, \overline{x'}p_{x'}]. \nonumber
\end{align}
By \cref{eq:choices}, all the six commutators (up to a sign) 
that appear in \cref{eq:bosonimplies} are zero.

For a large enough $m$, after a possible redefinition of the unit cell,
the movers $p_{x'}, p_{y'}$ will be supported 
only on two unit cells, one at the origin and another at $(1,0)$ or $(0,1)$.
Thus, there are only six $=\binom{4}{2}$ potentially nontrivial commutators 
among all the translates (with respect to $R'$) of $p_{x'}$ and $p_{y'}$,
but we have made these six commutators to vanish.
\end{proof}

\begin{lemma}\label{lem:goodbasis}
For any nontrivial topological charge $e$,
there exists a free basis for the stabilizer module (the columns of $\sigma$)
such that $e$ is represented by one basis element 
whose all components belong to the ideal $\mm = (x-1,y-1) \subset R$.
\end{lemma}
\begin{proof}
Let $\sigma$ be chosen such that $\epsilon = \sigma^\dagger \lambda_q$ satisfies \cref{eq:annmm}.
This is possible by \cite[\S7 Thm.~4]{Haah2013}.
Apply row operations $\GL(q;\FF_p)$ to $\epsilon$ so that $\epsilon|_{x=1,y=1}$ 
is in the reduced row echelon form.
The set of all topological charges is $\coker \epsilon$ which is in fact a vector space over $\FF_p$
on which the translation group has trivial action.
The dimension $k = \dim_{\FF_p} \coker \epsilon = \dim_{\FF_p} \coker \epsilon|_{x=1,y=1}$
is precisely the number of all zero rows of $\epsilon|_{x=1,y=1}$.
Hence, any nonzero element of the codomain of $\epsilon$ that is supported on these last $k$ components
represents a nontrivial topological charge.
Therefore, for any nontrivial charge $e$, 
its representative vector in $\coker \epsilon$ can be mapped
to a unit vector by some $\GL(k;\FF_p)$ acting on the last $k$ components.
\end{proof}

\begin{lemma}\label{lem:commutingmover=trivialmover}
Under a choice of a sufficiently large unit cell,
the movers $p_x,p_y$ for a nontrivial boson $e$ can be mapped to
\begin{align}
p_x &= \begin{pmatrix} 1 & 0 & 0 & 0 &\cdots & 0\end{pmatrix}^T, \nonumber\\
p_y &= \begin{pmatrix} 0 & 1 & 0 & 0 & \cdots & 0 \end{pmatrix}^T
\end{align}
by some elementary symplectic transformation $\ESp(q;R)$ (Clifford circuit).
\end{lemma}
\begin{proof}
By \cref{lem:isotropy-commutingmover},
we may assume that $p_x$ and $p_y$ are supported on at most two adjacent unit cells.
In particular, $p_x$ is a Laurent polynomial vector over $\FF_p[x^\pm]$,
satisfying $p_x^\dagger \lambda_q p_x = 0$.
Since~$\FF_p[x^\pm]$ is a principal ideal domain,
we can find an elementary symplectic transformation $E_1 \in \ESp(q;\FF_p[x^\pm])$ 
that turns~$p_x$ into a vector with a single nonzero component, say, $g$ at the first entry;
see computation in \cite[\S 6]{Haah2013}.
But the single component $g$ must be a monomial;
under the choice of $\epsilon$ of \cref{lem:goodbasis} 
we have $fg=x-1$ for some $f \in \mm = (x-1,y-1)$,
and $x-1$ being irreducible in a unique factorization domain $R$
forces $g$ to be a unit.

After the transformation $E_1$,
the $y$-mover becomes $E_1 p_y = h_1 y + h_0$ for some $h_1,h_0 \in \FF_p[x^\pm]^{2q}$.
To remove the appearance of $x$ in this $y$-mover,
we further take a coarse-graining $\cgx{m}$ induced by a ring injection 
$\chi^{(m)} : x' \mapsto x^m, y' \mapsto y$.
That is, $\cgx{m}$ is a coarse-graining along $x$-direction only.

Assuming $m$ is even, 
let us fix a specific basis of $R \cong R'^{m}$ as an $R'$-module:
$R = x^{-m/2}R' \oplus \cdots \oplus x^{m/2-1}R'$.
This choice is convenient because any Laurent polynomial vector $v \in R^{2q}$
becomes, by $\cgx{m}$, a vector with entries in $\FF_p[y]$ without $x$
for all sufficiently large $m$.

Once $x$-mover is put in the promised form,
any further coarse-graining $\cgx{m}$ does not complicate the $x$-mover;
if $p_x$ is supported on a single cell, i.e., $p_x$ involves no variable,
the operator $p_{x'} = \cg{m}((x^{-m/2}+x^{-m/2+1}+\cdots+ x^{m/2-1})p_x)$ 
is still supported on one new unit cell (no variables),
and $\ESp(mq;\FF_p)$, i.e., some Clifford transformation within a unit cell,
can bring $p_{x'}$ into the promised form.
This transformation of $\ESp(mq;\FF_p)$ is innocuous in that 
it does not insert or remove any variable to the $x$- and $y$-movers.
A benefit of $\cgx{m}$ is now that 
$p_{y'} = \cgx{m}(E_1 p_y)$ 
involves no $x'$ and furthermore $p_{y'} = t_1 y + t_0$ for some $t_1,t_0 \in \FF_p^{2mq}$.

Therefore, we may assume that $p_x$ is in the promised form 
with the sole nonzero entry $1$ in the first component,
and $p_y$ is a polynomial in $y$ with all exponents being $0$ or $1$ not involving $x$.
Since $p_x^\dagger \lambda_q p_y = 0$, the $q+1$-st component of $p_y$ must be zero.
This forces $q \ge 2$.
Now, we look at the $2,3,\ldots,q,q+2,q+3,\ldots,2q$-th components of $p_y$.
If they generate the unit ideal ($=\FF_p[y^\pm]$),
then clearly $p_y$ can be turned into the promised form.
If not, then by the exponent restriction we can use some transformation of $\ESp(q-1;\FF_p)$
to turn $p_y$ into one that has $y-v$ at the second component with $v \in \FF_p\setminus\{0\}$,
some $u \in \FF_p$ at the first component,
and zeros in all the other components.
Then, under the choice of $\epsilon$ of \cref{lem:goodbasis}
we would have $(x-1)u + f (y-v) = y-1$ for some $f \in \mm$,
but $ux - u - y+1 = -f(y-v)$ is an irreducible polynomial,
a contradiction.
\end{proof}

\begin{proof}[Proof of \cref{thm:poly}]
Let $k$ be the $\FF_p$-dimension of the torsion submodule of $\coker \epsilon$.
If $k=0$, then by \cite[Lem.~7.1]{Haah2013} $\sigma$ can be chosen to be kernel free,
and $\coker \epsilon$ must be pure torsion, which implies $\coker \epsilon = 0$,
so the first Fitting ideal of $\epsilon$ is unit, and \cite[Cor.~4.2]{Haah2013} says
that the code (ground space) on any finite periodic lattice encodes zero logical qudit.
Then, \cite[Thm.~IV.4]{HFH2018} provides a Clifford QCA that maps the stabilizer group to the trivial one.
The result of \cite{haah2019clifford} implies that this Clifford QCA is actually a Clifford circuit
that is translation invariant with respect to a translation subgroup of a finite index.

If $k \neq 0$,
as remarked earlier, we can assume \cref{eq:annmm}.
By \cref{lem:isotropic} we have a nontrivial boson, and its movers can be chosen as in \cref{lem:commutingmover=trivialmover} by reducing the translation invariance to a subgroup of a finite index.
Since these simplified movers have only one nonzero entry of $1$,
under the choice of stabilizer generators in \cref{lem:goodbasis}
we must have 
\begin{align}
\epsilon = \begin{pmatrix}
x - 1 & y-1 &  \star & \star & \cdots & \star \\
0     & 0   &  \star & \star & \cdots & \star \\
\vdots&\vdots& \vdots & \vdots & \cdots & \vdots \\
0     & 0   &  \star & \star & \cdots & \star 
\end{pmatrix} \label{eq:eps1}
\end{align}
where $\star$ indicates an unknown entry, but the first row has entries from $\mm = (x-1,y-1)$.
Since $\begin{pmatrix}
y-1 & -x+1 & 0 & \cdots & 0
\end{pmatrix}^T$
is in $\ker \epsilon$, 
the condition $\ker \epsilon = \im \lambda_q \epsilon^\dagger$ implies that
the rows of $\epsilon$ must generate
\begin{align}
\begin{pmatrix}
0 & \cdots & 0 & \bar y - 1 & -\bar x + 1 & 0 &\cdots & 0
\end{pmatrix} \label{eq:addlrow}
\end{align}
over $R$ where $\bar y -1$ is at the $q+1$-th position from the left.
But the $R$-linear combination that results in this row vector
cannot contain a nonzero summand of the first row in \cref{eq:eps1} 
because the base ring $R$ does not have any zero divisor.
This makes it impossible for $k$ to be $1$
since $k$ is the number of rows of $\epsilon$ that becomes zero by setting $x=1=y$;
see the proof of \cref{lem:goodbasis}.

Let us extend $\epsilon$ to include a row \cref{eq:addlrow} to make $\epsilon'$;
this amounts to increasing the number of generators for the stabilizer module by 1.
(The new $\epsilon'$ does not satisfy \cref{eq:annmm} since $\coker \epsilon'$ is not torsion.)
\begin{align}
\epsilon' = \begin{pmatrix}
x - 1 & y-1 &  \star & \cdots & \star & u' & v' & \star &\cdots & \star \\
0 & 0 &        0 &\cdots & 0& \bar y - 1 & -\bar x + 1 & 0 &\cdots & 0\\
0     & 0   &  \star & \cdots &\star & \star & \star & \star & \cdots & \star\\
\vdots&\vdots& \vdots & \cdots &\vdots & \vdots & \vdots & \vdots & \cdots& \vdots\\
0     & 0   &  \star & \cdots &\star & \star & \star & \star & \cdots & \star
\end{pmatrix} \label{eq:eps2}
\end{align}
where $u',v'$ are unknown.
Since the first row of $\epsilon'$ has entries in $\mm$,
the two generators of $\mm$ in the top left can eliminate all $\star$ in the first row
by control-Not and Hadamard gates which act on the right of $\epsilon'$.
The second row also remains intact. Thus we obtain
\begin{align}
\epsilon'' = \begin{pmatrix}
x - 1 & y-1 &  0 & \cdots & 0 & u'' & v'' & 0 &\cdots & 0 \\
0 & 0 &        0 &\cdots & 0& \bar y - 1 & -\bar x + 1 & 0 &\cdots & 0\\
0     & 0   &  \star & \cdots &\star & \star & \star & \star & \cdots & \star\\
\vdots&\vdots& \vdots & \cdots &\vdots & \vdots & \vdots & \vdots & \cdots& \vdots\\
0     & 0   &  \star & \cdots &\star & \star & \star & \star & \cdots & \star
\end{pmatrix}. \label{eq:eps3}
\end{align}
Here, each pair of entries in the columns of $u'',v''$, below the second row, must be a $R$-multiple of
$\begin{pmatrix}
\bar y - 1 & -\bar x +1
\end{pmatrix}$
because $\epsilon''^\dagger \lambda_q \epsilon'' = 0$;
the symplectic product with the first row enforces this.
Hence, they can be eliminated by $\GL(q+1;R)$ on the left of $\epsilon'$.
Now we use the following elementary fact which will be proved shortly.
\begin{lemma}\label{lem:xyab}
For $u,v \in R$, if a matrix 
\begin{align}
M = 
\begin{pmatrix}
x-1 & y-1 & u & v \\
0   & 0   & \bar y -1 & -\bar x + 1
\end{pmatrix}
\end{align}
satisfies $M \lambda_2 M^\dagger = 0$,
then 
there exist $A \in \GL(2;R)$ and $B \in \ESp(2;R)$
such that $A M B = \epsilon_0$ where $\epsilon_0$ is defined in \cref{eq:eps0}.
\end{lemma}

Then we obtain
\begin{align}
\epsilon''' = \begin{pmatrix}
x - 1 & y-1 &  0 & \cdots & 0 & 0 & 0 & 0 &\cdots & 0 \\
0 & 0 &        0 &\cdots & 0& \bar y - 1 & -\bar x + 1 & 0 &\cdots & 0\\
0     & 0   &  \star & \cdots &\star & 0 & 0 & \star & \cdots & \star\\
\vdots&\vdots& \vdots & \cdots &\vdots & \vdots & \vdots & \vdots & \cdots& \vdots\\
0     & 0   &  \star & \cdots &\star & 0 & 0 & \star & \cdots & \star
\end{pmatrix}. \label{eq:eps4}
\end{align}
It is evident that $\epsilon_0$ is a direct summand.
The change $\epsilon \to \epsilon' \to \epsilon'' \to \epsilon'''$
is achieved by $\ESp(q;R)$ and change of a generating set of the stabilizer module.
The complement of $\epsilon_0$ in $\epsilon'''$
gives a stabilizer module that satisfies the conditions of our theorem.
This proves the induction step decreasing $k$ by 2,
and thus completes the proof of the theorem.
\end{proof}

\begin{proof}[Proof of \cref{lem:xyab}]
By long division we can write $u = (\bar y - 1)u' + u''$ where $u' \in R, u'' \in \FF_p[x^\pm]$.
By a row operation on the left of $M$, we can eliminate $(\bar y -1)u'$
and thus we assume from now on that $u = u'' \in \FF_p[x^\pm]$.
The equation $M \lambda_2 M^\dagger =0$ implies that
$
(x-1) \bar u + (y-1) \bar v = (\bar x -1) u + (\bar y -1) v
$, which can be rearranged as
\[
(x-1)(\bar u + \bar x u) = -(y-1) (\bar v + \bar y v).
\]
The left hand side is a Laurent polynomial in $x$,
and therefore, if nonzero, it is not divisible by $y-1$.
Hence, $\bar u + \bar x u = 0 = \bar v + \bar y v $.
Let $u = \sum_j u_j x^j$ where $u_j \in \FF_p$ be the expansion of $u$.
It follows that $u_{-j} + u_{j+1} = 0$ for all $j$.
This implies that $u|_{x=1} = 0$ so $u = (x-1)h$ for some $h \in \FF_p[x^\pm]$.
Substituting, we have $ \bar u + \bar x u = (\bar x -1 )\bar h - (\bar x -1) h = 0$ or $h = \bar h$.
Thus the control-Phase gate on the first qudit can eliminate $u$.

We are left with an equation $\bar v + \bar y v = 0$.
Write $v = (x-1)s + f$ where $f \in \FF_p[y^\pm]$.
Then, we have $(x-1)(-\bar x \bar s + \bar y s) + \bar f + \bar y f = 0$.
Since $f$ is constant in $x$, we have $f + \bar y f = 0$ 
implying, as before, that $f = (y-1)g$ for some $g = \bar g \in \FF_p[y^\pm]$,
and hence $f$ can be eliminated by the control-Phase on the second qudit.

The remaining term $(x-1)s$ of $v$ satisfies $\bar x y \bar s = s$.
We claim that $s = y \bar r + \bar x r$ for some $r \in R$,
which will conclude the proof:
\begin{align}
\begin{pmatrix}
1 & -r \\ 0 & 1
\end{pmatrix}
\begin{pmatrix}
x-1 & y-1 & 0 & (x-1)s \\
0   & 0   & \bar y -1 & -\bar x + 1
\end{pmatrix}
\begin{pmatrix}
1 & 0 & 0 & -y \bar r\\
0 & 1 & -\bar y r & 0\\
0 & 0 & 1 & 0 \\
0 & 0 & 0 & 1
\end{pmatrix}
=
\begin{pmatrix}
x-1 & y-1 & 0 & 0 \\
0   & 0   & \bar y -1 & -\bar x + 1
\end{pmatrix}.
\end{align}
To show the claim,
let $z = \bar x y$.
Write $s = \sum_j y^j s_j$ where $s_j \in \FF_p[z^\pm]$;
such an expression is unique as seen by considering a ring isomorphism 
$\FF_p[z^\pm,y^\pm] \to \FF_p[x^\pm,y^\pm]$
where $z \mapsto \bar x y, y \mapsto y$.
We have $s_{-j} = z \bar s_j$ for all $j$ and hence $s_0 = (z+1)\ell$ 
for some $\ell = \bar \ell \in \FF_p[z^\pm]$.
Let $s_+ = \sum_{j > 0} y^j s_j$.
Then $s = s_+  + (z+1)\ell + z \overline{ s_+} = y \overline{(x \ell + x s_+)} + \bar x (x \ell + x s_+)$.
\end{proof}

\begin{acknowledgments}
I thank Lukasz Fidkowski and Matt Hastings for fruitful discussions 
and Zhenghan Wang for informing me of \cite{Galindo2016}.
I also thank an anonymous referee for comments 
that helped improving the clarity of the manuscript.
\end{acknowledgments}

\bibliography{cph2d-ref}
\end{document}